\documentclass[10pt]{article}
\usepackage{amsmath}
\usepackage{amsfonts}
\usepackage{amssymb}
\usepackage{amscd}
\usepackage{amsthm}
\usepackage{latexsym}
\usepackage{mathrsfs}
\usepackage[all]{xy}
\usepackage{fullpage}

\newtheorem{theorem}{Theorem}
\newtheorem{proposition}{Proposition}
\newtheorem{lemma}{Lemma}

\theoremstyle{definition}

\newcommand{\TS}{\mathsf{TS}}
\newcommand{\AS}{\mathsf{AS}}
\newcommand{\KAT}{\mathsf{KAT}}
\newcommand{\KAD}{\mathsf{KAD}}
\newcommand{\PHL}{\mathsf{PHL}}
\newcommand{\aran}{\mathit{ar}}

\begin{document}

\title{On the Expressive Power of Kleene Algebra with Domain}

\author{Georg Struth\\University of Sheffield, UK}

\date{}

\maketitle 

\begin{abstract}
  It is shown that antidomain semirings are more expressive than test
  semirings and that Kleene algebras with domain are more expressive
  than Kleene algebras with tests. It is also shown that  Kleene
  algebras with domain are expressive for propositional Hoare logic
  whereas Kleene algebras with tests are not.
\end{abstract}

\pagestyle{plain}


\section{Introduction}

Kleene algebras with tests ($\KAT$)~\cite{Kozen97} yield arguably the
simplest and most elegant model of the control flow in simple
while-programs. They provide an abstract algebraic view on the
standard relational semantics of imperative programs, have been
applied to various program analysis tasks and form the backbone of
program construction and verification tools. In particular, the
inference rules of propositional Hoare logic ($\PHL$)---Hoare logic
without assignment rule---can be derived in this
setting~\cite{Kozen00}.  Kleene algebras with domain
($\KAD$)~\cite{DesharnaisMS06,DesharnaisStruth} are a similar
formalism that provides an algebraic approach to propositional dynamic
logic and predicate transformer semantics. The inference rules of
$\PHL$ are derivable in $\KAD$ as well and it is known that
every $\KAD$ is a $\KAT$~\cite{DesharnaisStruth}.

From a complexity point of view, the equational theory of $\KAT$ is
known to be PSPACE complete~\cite{KozenS96}, whereas that of $\KAD$ is
decidable in EXPTIME~\cite{MoellerS06}. It seems also plausible that
$\KAD$ is more expressive than $\KAT$; after all, image and preimage
as well as modal box and diamond operators can be defined in the
former algebra. 

This article makes this gap in expressive power precise, showing that
$\KAD$ is strictly more expressive than $\KAT$ with a
simple, natural and interesting example.  Firstly it is shown that the
inverse of the sequential composition rule of $\PHL$, when expressed
as a formula in the language of $\KAT$, is derivable from the axioms
of $\KAD$.  Secondly, a model of $\KAT$ is presented in which this
formula does not hold. In addition it is shown that $\KAT$ is not
expressive for $\PHL$, whereas this is trivially the case for $\KAD$.

Inverting the inference rules of Hoare logic is interesting for
verification condition generation in the context of program
correctness, where intermediate assertions such as weakest liberal
preconditions need to be computed. It is also related to the question
of expressivity of Hoare logic in relative completeness proofs.


\section{$\KAD$ and $\KAT$}

A \emph{semiring} is a structure $(S,+,\cdot,0,1)$ such that $(S,+,0)$
is a commutative monoid, $(S,\cdot, 1)$ is a monoid; and the two
monoids interact via the distributivity laws $x\cdot (y+z)=x \cdot y+
x\cdot z$ and $(x+y)\cdot z= x\cdot z + y\cdot z$ and the
annihilation laws $0\cdot x=0$ and $x\cdot 0=0$.  

A \emph{dioid} is an additively idempotent semiring, that is, $x+x=x$
holds for all $x\in S$. In this case, $(S,+)$ forms a semilattice with
order relation defined as $x\le y\Leftrightarrow
x+y=y$. Multiplication is isotone with respect to the order, $x\le y$
implies both $z\cdot x\le z\cdot y$ and $x\cdot z\le y\cdot z$, and
$0\le x$ holds for all $x\in S$.

A \emph{Kleene algebra} is a dioid expanded by a 
star operation that satisfies the unfold and induction axioms 
\begin{equation*}
  1+x\cdot x^\ast = x^\ast,\qquad 1+ x^\ast \cdot x = x^\ast,\qquad z+x\cdot y \le y \Rightarrow x^\ast 
  \cdot z \le y,\qquad z + y \cdot x \le y \Rightarrow z\cdot x^\ast
  \le y.
\end{equation*}

An \emph{antidomain semiring}~\cite{DesharnaisStruth} is a semiring
$S$ endowed with an operation $a:S\to S$ that satisfies
\begin{equation*}
  a(x)\cdot x = 0,\qquad
a(x\cdot y)+ a (x \cdot a (a(y))) = a(x\cdot a(a(y)),\qquad
a(x)+a(a(x)) = 1.
\end{equation*}
These axioms imply that every antidomain semiring is a dioid.  A
\emph{domain operation} can be defined on $S$ as $d=a\circ a$. It is a
retraction, that is, $d\circ d= d$, and it follows that $x\in
d(S)\Leftrightarrow d(x)=x$, where $d(S)$ denotes the image of the set
$S$ under $d$. This fact can be used to show that
$(d(S),+,\cdot,a,0,1)$ forms a boolean algebra in which multiplication
coincides with meet and the antidomain operator $a$ yields test
complementation.  In addition we need the following fact about
antidomain semirings.
\begin{lemma}[\cite{DesharnaisStruth}]\label{P:weak-loc}
  In every antidomain semiring, $ x\cdot y = 0\Leftrightarrow x\cdot
  d(y)=0$.
\end{lemma} 
\noindent A \emph{Kleene algebra with domain}~\cite{DesharnaisStruth} is both a
Kleene algebra and an antidomain semiring.

A \emph{test semiring} is a dioid $S$ in which a boolean algebra $B$
is embedded by a map $\iota:B\to S$ such that
\begin{equation*}
  \iota(0)=0,\qquad \iota(1)=1,\qquad \iota (x\sqcup y)=
  \iota(x)+\iota(y),\qquad \iota(x\sqcap y) = \iota(x)\cdot \iota(y).
\end{equation*}
A \emph{Kleene algebra with tests}~\cite{Kozen97} is both a Kleene
algebra and a test semiring.  In the tradition of Kleene algebras with
tests the embedding is left implicit. I write $p,q,r,\dots$ for
boolean elements, which are called \emph{tests}, and $x,y,z$ for
arbitrary semiring elements. I write $\AS$ for the class and axiom
system of domain semirings, $\TS$ for that of test semirings, $\KAD$
for that of Kleene algebras with domain and $\KAT$ for that of Kleene
algebras with tests.

\begin{lemma}[\cite{DesharnaisStruth}]
 $\KAD\subseteq \KAT$.
\end{lemma}
\begin{proof}
  If $K\in\KAD$ then $d(K)$ is a boolean algebra, hence a test
  algebra. The embedding is provided by the identity function on
  $d(K)$ as a subset of $K$. Thus $K\in\KAT$.
\end{proof}
It follows that $\AS\subseteq \TS$. Thus, for any $K\in\KAD$, all
elements in $d(K)$  may serve as tests in the associated
$(K,d(K))\in\KAT$.

The notions of domain, antidomain and tests can be motivated from the
model of binary relations.

\begin{proposition}[\cite{Kozen00,DesharnaisStruth}]
  Let $2^{A\times A}$ be the set of binary relations over the set
  $A$. Suppose that 
  \begin{gather*}
    R\cdot S =\{(a,b)\mid \exists c.\ (a,c)\in R \wedge (c,b)\in S\},\qquad
   \mathit{id}=\{(a,a)\mid a\in A\},\\
a(R) = \{(a,a)\mid \forall b.\  (a,b)\not\in R\},\qquad
R^\ast = \bigcup_{i\in\mathbb{N}} R^i,
  \end{gather*}
where $R^0=\mathit{id}$ and $R^{i+1}= R\cdot R^i$.  Then
\begin{enumerate}
\item $(2^{A\times A},\{R\mid R\subseteq
  \mathit{id}\},\cup,\cdot,\emptyset,\mathit{id},^\ast)\in \KAT$,
\item $(2^{A\times
    A},\cup,\cdot,\emptyset,\mathit{id},a,^\ast\}\in\KAD$.
\end{enumerate}
\end{proposition}
The operation $\cdot$ on relations is the standard relational product;
$\mathit{id}$ is the identity relation on $S$. The operation $a$ is
the domain complement on relations; $a(R)$ represents those states in
$S$ that are not related by $R$ to any other state. 


\section{Expressive Power of $\KAD$ and Invertibility in $\PHL$}

To show that domain semirings are strictly more expressive than test
semirings and that Kleene algebras with domain are strictly more
expressive than Kleene algebras with tests I display a 
sentence $\varphi$ in the language of $\KAT$ such that
$\KAT\not\vdash\varphi$ and $\KAD\vdash\varphi$. To prove
that $\KAT\not\vdash\varphi$ I display a $(K,B)\in\KAT$ such that
$(K,B)\not\models\varphi$.

The sentence $\varphi$ chosen for this purpose is related to the
relative completeness of Hoare logic. It is well known that the
validity of a Hoare triple can be encoded in the language of $\KAT$~\cite{Kozen00},
and hence $\KAD$, as
\begin{equation*}
\{p\}x\{q\}
\Leftrightarrow p\cdot x \cdot \overline{q}=0,
\end{equation*}
where tests $p$ and $q$ serve as assertions and $\overline{q}$
represents the boolean complement of test $q$. Moreover the inference
rules of $\PHL$ are derivable in $\KAT$~\cite{Kozen00}. In particular
the rule $\{p\}x\{r\}\wedge \{r\}y\{q\}\Rightarrow \{p\}x\cdot y\{q\}$
for sequential composition can be derived in $\TS$ and $\AS$.
Invertibility of this rule means finding for any Hoare triple
$\{p\}x\cdot y\{q\}$ an assertion $r$ such that $\{p\}x\{r\}$ and
$\{r\}y\{q\}$. Hence consider the following sentence in the language
of $\KAT$:
\begin{equation*}
 \varphi \ \equiv \ (\forall x,y\in K, p\in B, q\in B.\ \{p\} x\cdot
 y\{q\}\  \Rightarrow \ (\exists r\in B.\ \{p\} x\{r\}\wedge \{r \} y \{q\})).
\end{equation*}
\begin{lemma}\label{P:katisbad}
  $\KAT\not\vdash\varphi$.
\end{lemma}
\begin{proof}
  Consider the $\KAT$ $(\{a\},\{0,1\},+,\cdot,0,1,^\ast)$ with
  addition defined by $0\le a\le 1$, multiplication by $a\cdot a = 0$
  and $a^\ast = 1$ (all other operations on elements being fixed).
  Note that $a$ is not a test because $a\cdot a\neq a$. In this
  algebra, $1 \cdot a \cdot a \cdot \overline{0} = 1 \cdot 0 \cdot
  1=0$. However, $r$ can neither be $0$ or $1$. In the first case,
  $1\cdot a\cdot \overline{0}=1\cdot a \cdot 1= 1$; in the second one,
  $1\cdot a\cdot\overline{0}=1\cdot a\cdot 1=a$.
\end{proof}

\begin{lemma}\label{P:kadisgood}
  $\AS\vdash\varphi$.
\end{lemma}
\begin{proof}
  Let $S\in\AS$ and suppose $p\cdot x \cdot y \cdot \overline{q}=0$,
  with $p,q\in d(S)$. We need an expression $r$ such that $p\cdot
  x\cdot\overline{r}=0$ and $r\cdot y\cdot \overline{q}=0$. So let $r=
  a(y\cdot \overline{q})$. The assumption  and Lemma~\ref{P:weak-loc}
  then imply that $p\cdot
  x\cdot \overline{r} = p \cdot x\cdot d(y\cdot
  \overline{q})=0$. Moreover, $r \cdot y\cdot\overline{q} = a (y\cdot
  \overline{q}) \cdot y \cdot \overline{q} =0$ follows from the first
  antidomain axiom.
\end{proof}
\noindent These two lemmas can be summarised as follows.
\begin{theorem}\label{P:main}
  There exists a sentence in the language of $\KAT$ which
  is derivable from the $\AS$ axioms, but not from the $\KAT$ axioms.
\end{theorem}
Thus antidomain semirings are strictly more expressive than test
semirings, and  Kleene algebras with domain are strictly more
expressive than Kleene algebras with tests.


\section{Expressive Power of $\KAD$ and Expressivity of $\PHL$}

The question of invertibility of the rules of Hoare logic relates to
its expressivity, requiring that for each command $x$ and
postcondition $q$ the weakest liberal precondition be definable. In
any $K\in \KAD$, the weakest liberal precondition
exists for any element $x\in K$ and test $p\in d(K)$ by definition.

Formally, for all $x,y\in K$  one can define a modal box operator
\begin{equation*}
  [x]y = a(x \cdot a(y))
\end{equation*}
and show that $p\le [x]q \Leftrightarrow p\cdot x\cdot q = 0$. So
$\{p\}x\{p\} \Leftrightarrow p\le [x]q$ yields an alternative
definition of the validity of Hoare triples, in which $\lambda p.\
[x]p:d(K)\to d(K)$ is a predicate transformer~\cite{MoellerS06}.

It follows that $\{[x]q\}x\{q\}$---$]x]p$ is a precondition for $x$
and $q$---and $\{p\}x\{q\}\Rightarrow p\le [x]q$---$[x]p$ is weaker
than any other precondition of $x$ and $q$.  Hence $[x]q$ models
indeed the weakest liberal precondition of $x$ and $q$. Since the
standard relational semantics of while programs withouth the
assignment rules can be captured in $\KAT$~\cite{Kozen97} (and $\KAD$)
by defining $\mathbf{if}\ p\ \mathbf{then}\ x\ \mathbf{else} \ y =
p\cdot x+\overline{p}\cdot y$ and $\mathbf{while}\ p\ \mathbf{do}\ x =
(p\cdot x)^\ast \cdot \overline{p}$, the following fact is obvious.
\begin{theorem}
 $\KAD$ is expressive for $\PHL$. 
\end{theorem}

The proof of Lemma~\ref{P:kadisgood} can now be rewritten in the light
of this discussion.  First of all, $r=[y]q$ models precisely the
weakest liberal precondition of $y$ and $q$.  The next Lemma then
arises as an instance of $\varphi$ in combination with the sequential
composition rule of Hoare logic.
\begin{lemma}
$\AS\vdash  \{p\}x\cdot y\{q\} \Leftrightarrow \{p\}x\{[y]q\}$.
\end{lemma}
\noindent The second, implicit  conjunct is of course $\{[x]q\}x\{q\}$. It is
valid and has therefore been deleted. 

In $\KAT$ the situation is different.
\begin{theorem}
$\KAT$ is not expressive for $\PHL$.
\end{theorem}
\begin{proof}
  Let $A$ be an infinite set and $\mathscr{B}=\{B\subseteq A \mid B
  \text{ is finite}\}\cup\{B\subseteq A\mid B \text{ is
    cofinite}\}$. It has been shown that
  $(2^A,\mathscr{B},\cup,\cap,\emptyset,S,^\ast)\in \KAT$ in which
  $B^\ast = A$ for all $B\subseteq
  A$~\cite{DesharnaisMS06}. The test algebra $\mathscr{B}$ is
  not complete because suprema of infinitely many finite sets need not
  be in $\mathscr{B}$.

  Consider the set $C\in 2^A-\mathscr{B}$ and suppose that $a(C\cap
  a(\emptyset))= a(C\cap A)=a(C)$, the weakest liberal precondition of
  $C$ and $\emptyset$, exists. Thus $a(C)\cap C=\emptyset$ by
  definition. In addition, $C$ has of course a complement
  $\overline{C}\in A-\mathscr{B}$ as well. It follows that
  $a(C)\subset\overline{C}$ and hence $\overline{C}-a(C)\neq\emptyset$.

  So let $x\in \overline{C}-a(C)$ and consider the set
  $a(C)\cup\{x\}$.  By construction it is an element of $\mathscr{B}$
  that contains $a(C)$ and still satisfies $(a(C)\cup\{x\})\cap
  C=\emptyset$. This contradicts the maximality assumption on $a(C)$.

Hence there is a $\KAT$ in which for some element $x$ and test $p$ the
weakest liberal precondition $[x]p$ of $x$ and $p$ does not exist and
$\KAT$ is not expressive for $\PHL$.
\end{proof}


\section{Concluding Remarks} 

The left distributivity law $x\cdot (y+z)= x\cdot z + y \cdot z$ is
not needed in the proof of Lemma~\ref{P:kadisgood} (and
Lemma~\ref{P:weak-loc}). Formula $\varphi$ can be derived
already from the axioms of antidomain
near-semirings~\cite{DesharnaisS08} and Kleene algebras with
domain based on near-semirings are already more expressive than
$\KAT$.

The result of Lemma~\ref{P:kadisgood} can be dualised and extended, so
that other solutions for $r$ can be found. A notion of opposition
duality can be defined on a semiring by swapping the order of
multiplication. Obviously, the opposite of every Kleene algebra is
again a Kleene algebra.  The domain operation on a semiring translates
to a range operation on the opposite semiring, and vice
versa~\cite{DesharnaisStruth}.  Thus an antirange and a range
operation on a semiring can be axiomatised by $x\cdot \aran(x)= 0$,
$\aran(x\cdot y)+ \aran(r(x)\cdot y)= \aran(r(x)\cdot y)$ and
$\aran(x)+r(x)=1$. It is then easy to check that $r=r(p\cdot x)$
provides a solution to a dual variant of Lemma~\ref{P:kadisgood}. The
proof uses the fact that $x\cdot y=0$ is equivalent to $r(x)\cdot y=0$
in antirange semirings, which is obtained from Lemma~\ref{P:weak-loc}
by opposition duality.

One can also consider Kleene algebras with antidomain and antirange
operations.  It is then appropriate to impose $d(\aran(x))=\aran(x)$
and $r(a(x))= a(x)$ to enforce that $d(S)$ and $r(S)$
coincide~\cite{DesharnaisStruth}.  In this context, also $r=r(p\cdot
x)\cdot a(y\cdot \overline{q})$ provides a  third solution to a generalised
variant of Lemma~\ref{P:kadisgood}.

A final remark concerns the invertibility of the remaining inference
rules of propositional Hoare logic.  Invertibility of  the consequence
rule(s) is trivial.  The equivalence
\begin{equation*}
    \{p\cdot t\}x\{q\}\wedge \{p\cdot\overline{t}\}y\{q\}\Leftrightarrow \{p\}\mathbf{if}\ p\
  \mathbf{then}\ x\ \mathbf{else}\ y\{q\}
\end{equation*}
is derivable in $\KAT$: that $\{p\}\mathbf{if}\ p\ \mathbf{then}\ x\
\mathbf{else}\ y\{q\}$ implies $\{p\cdot t\}x\{q\}$, for instance, is
verified by
  \begin{equation*}
    0=t\cdot 0 = t\cdot p\cdot (t \cdot 
  x + \overline{t} \cdot y)\cdot \overline{q}= (p\cdot t\cdot t \cdot 
  x\cdot \overline{q}
  + p \cdot t \cdot \overline{t}\cdot y\cdot \overline{q} = p\cdot t 
  \cdot x\cdot \overline{q}+ 0= \{p\cdot t\}x\{q\}. 
  \end{equation*}
For the while rule $\{p\cdot t\}x\{p\}\Rightarrow
\{p\}\mathbf{while}\ t\ \mathbf{do}\ x\{p\cdot\overline{t}\}$, the
stronger consequent $\{p\}(t\cdot x)^\ast\{p\}$---the while loop
satisfies the invariant $p$---is derivable from the antecedent, and
invertibility follows from
\begin{equation*}
  p\cdot t \cdot x \cdot \overline{p}\le p\cdot (t\cdot x)^\ast \cdot \overline{p}=0.
\end{equation*}

\bibliographystyle{plain}
\bibliography{kad_expressivity}

\begin{thebibliography}{1}

\bibitem{DesharnaisMS06}
J.~Desharnais, B.~M{\"o}ller, and G.~Struth.
\newblock Kleene algebra with domain.
\newblock {\em ACM TOCL.}, 7(4):798--833, 2006.

\bibitem{DesharnaisS08}
J.~Desharnais and G.~Struth.
\newblock Domain axioms for a family of near-semirings.
\newblock In J.~Meseguer and G.~Rosu, editors, {\em AMAST 08}, volume 5140 of
  {\em LNCS}, pages 330--345. Springer, 2008.

\bibitem{DesharnaisStruth}
J.~Desharnais and G.~Struth.
\newblock Internal axioms for domain semirings.
\newblock {\em Science of Computer Programming}, 76(3):181--203, 2011.

\bibitem{Kozen97}
D.~Kozen.
\newblock Kleene algebra with tests.
\newblock {\em ACM TOPLAS}, 19(3):427--443, 1997.

\bibitem{Kozen00}
D.~Kozen.
\newblock On {H}oare logic and {K}leene algebra with tests.
\newblock {\em ACM TOCL}, 1(1):60--76, 2000.

\bibitem{KozenS96}
D.~Kozen and F.~Smith.
\newblock Kleene algebra with tests: Completeness and decidability.
\newblock In D.~van Dalen and M.~Bezem, editors, {\em CSL'96}, volume 1258 of
  {\em LNCS}, pages 244--259. Springer, 1997.

\bibitem{MoellerS06}
B.~M{\"o}ller and G.~Struth.
\newblock Algebras of modal operators and partial correctness.
\newblock {\em Theoretical Computer Science}, 351(2):221--239, 2006.

\end{thebibliography}

\end{document}